\documentclass[11pt,letter]{article}
\usepackage{authblk}
\usepackage[margin=1in]{geometry}
\usepackage[utf8x]{inputenc}
\usepackage{amsmath, amsthm}
\usepackage{wrapfig,floatflt,graphicx,amssymb,textcomp,array,amsmath}
\usepackage{enumerate}
\usepackage{multirow}
\usepackage{tabularx}
\usepackage{color}
\usepackage{todonotes}
\usepackage[titletoc,title]{appendix}

\newcommand{\etal}{{et~al.}}

\newcommand{\ray}[1]{\overrightarrow{#1}}

\title{A Short Proof of the Toughness of Delaunay Triangulations
}

\author{Ahmad Biniaz\thanks{Part of this work has been done while the author was an NSERC postdoctoral fellow at University of Waterloo.}
}

\affil{School of Computer Science\\University of Windsor\\\texttt{ahmad.biniaz@gmail.com}}

\date{}

\newtheorem{conjecture}{Conjecture}
\newtheorem{theorem}{Theorem}

\newtheorem*{problem*}{Problem}
\newtheorem*{claim*}{Claim}
\newtheorem*{invariant*}{Invariant}

\begin{document}
	\maketitle
	\begin{abstract}
	We present a self-contained short proof of the seminal result of Dillencourt (SoCG 1987 and DCG 1990) that Delaunay triangulations, of planar point sets in general position, are 1-tough. An important implication of this result is that Delaunay triangulations have perfect matchings. Another implication of our result is a proof of the conjecture of Aichholzer~\etal~(2010) that at least $n$ points are required to block any $n$-vertex Delaunay triangulation.
	\end{abstract}
\section{Introduction}

Let $P$ be a set of points in the plane that is in general position, i.e., no three points on a line and no four points on a circle. The {\em Delaunay triangulation} of $P$ is an embedded planar graph with vertex set $P$ that has a straight-line edge between two points $p,q\in P$ if and only if there exists a closed disk that has only $p$ and $q$ on its boundary and does not contain any other point of $P$.
A graph is $1$-{\em tough} if for any $k$, the removal of $k$ vertices splits the graph
into at most $k$ connected components. In 1987, Dillencourt proved the toughness of Delaunay triangulations. 

\begin{theorem}[Dillencourt \cite{Dillencourt1990}]
	\label{Dillencourt-thr}
	Let $T$ be the Delaunay triangulation of a set of points in the plane in general position, and let $S\subseteq V(T)$. Then $T\setminus S$ has at most $|S|$ components. 
\end{theorem} 

Dillencourt's proof of Theorem~\ref{Dillencourt-thr} is nontrivial and employs a large set of combinatorial and structural properties of (Delaunay) triangulations. Using the same proof idea, he showed that if $T$ is a Delaunay triangulation of an arbitrary point set in the plane (not necessarily in general position) then $T\setminus S$ has at most $|S|+1$ components. Combining this with Tutte's classical theorem that characterizes graphs with perfect matchings \cite{Tutte1947}, implies the following well-known result.

\begin{theorem}[Dillencourt \cite{Dillencourt1990}]
	\label{matching-thr}
	Every Delaunay triangulation has a perfect matching.
\end{theorem}
In this note we present a self-contained short proof of Theorem~\ref{Dillencourt-thr}. To that end, we first present an upper bound on the maximum size of an independent set of $T$. To facilitate comparisons we use the same definitions and notations as in \cite{Dillencourt1990}. The number of elements of a set $S$ is denoted by $|S|$. For a graph $G$, the vertex set of $G$ is denoted by $V(G)$, and $|G|=|V(G)|$. 

Every interior face of $T$ is
a triangle, and the boundary of $T$ is a convex polygon; see Figure~\ref{DT-fig}(a).
An edge is called a {\em boundary edge} if it is on the boundary of $T$, and is called an {\em interior edge} otherwise.
For any interior edge $(p,q)\in T$ between two faces $pqr$ and $pqs$ it holds that 
\begin{equation}
\label{eq0}
\angle prq + \angle psq < 180.
\end{equation}

\begin{figure}[htb]
	\centering
	\setlength{\tabcolsep}{0in}
	$\begin{tabular}{ccc}
	\multicolumn{1}{m{.24\columnwidth}}{\centering\includegraphics[width=.2\columnwidth]{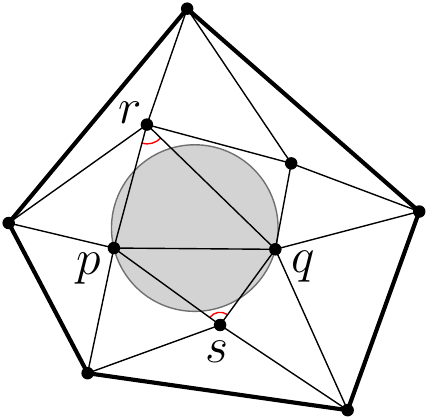}}
	&\multicolumn{1}{m{.52\columnwidth}}{\centering\includegraphics[width=.51\columnwidth]{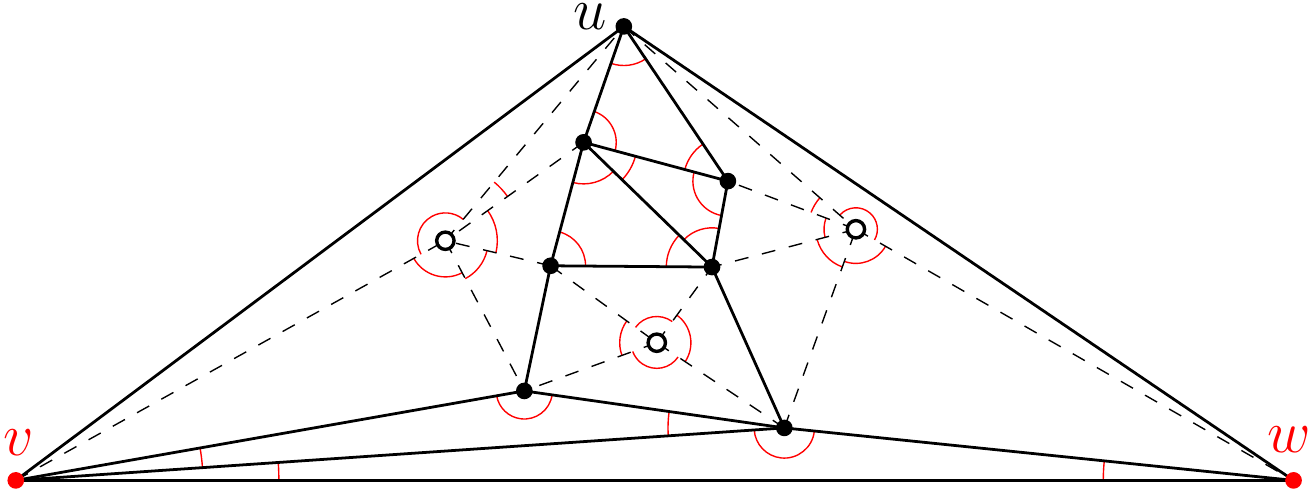}}&\multicolumn{1}{m{.24\columnwidth}}{\centering\includegraphics[width=.22\columnwidth]{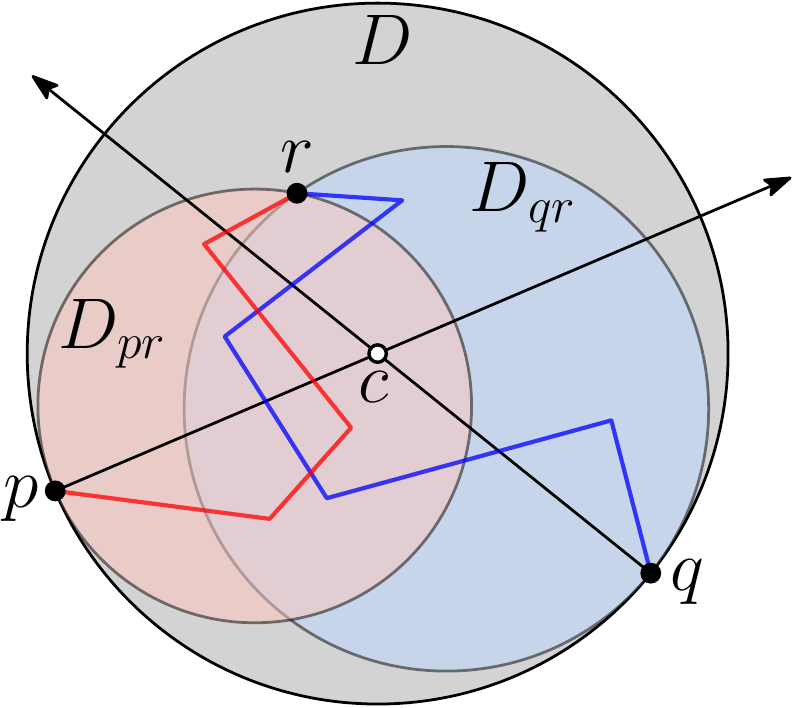}}\\
	(a)&(b)&(c)
	\end{tabular}$
	\caption{(a) The Delaunay triangulation $T$; bold segments are boundary edges. (b) The Delaunay triangulation $\mathcal{T}$; white vertices belong to $I$, solid edges belong to $\mathcal{T}[\mathcal{S}]$, and marked angles are distinguished angles. (c) Illustration of the proof of Theorem~\ref{Delaunay-thr}.}
	\label{DT-fig}
\end{figure}

\section{A combinatorial and a structural property}
\label{property-section}
A direct implication of Theorem~\ref{Dillencourt-thr} gives the upper bound $\lfloor(|T|+1)/2\rfloor$ on the size of any independent set of $T$; see e.g. \cite{Aichholzer2013}. We present a different self-contained proof for a slightly better bound.   

\begin{theorem}
\label{independent-set-thr}
Let $T$ be the Delaunay triangulation of a set of points in the plane in general position, and let $I$ be an independent set of $T$. Then $|I|\leqslant \lfloor|T|/2\rfloor$, and this bound is tight. 
\end{theorem}

\begin{proof}
This upper bound is tight as any maximum independent set in the $n$-vertex Delaunay triangulation of Figure~\ref{sufficiency-fig}(b) has exactly $\lfloor n/2\rfloor$ vertices (regardless of the parity of $n$).
	
Now we prove the upper bound. Set $S:=V(T)\setminus I$, and let $u$ be a vertex of $S$ that is on the boundary of $T$ (observe that such a vertex exists). Let $v,w\notin V(T)$ be two points in the plane such that (i) $T$ lies in the triangle $(u,v,w)$ and (ii) neither of $v$ and $w$ lies in the disks that introduce edges of $T$; see Figure~\ref{DT-fig}(b). Let $\mathcal{T}$ be the Delaunay triangulation of $V(T)\cup \{v,w\}$. Our choice of $v$ and $w$ ensures that any edge of $T$ is also an edge of $\mathcal{T}$, and thus $T\subset \mathcal{T}$. 
Set $\mathcal{S}:=S\cup \{v,w\}$.
In the rest of the proof we show that $|I|\leqslant |\mathcal{S}|-2$. This implies that $|I|\leqslant |S|$ (because $|\mathcal{S}|=|S|+2$) which in turn implies that $|I|\leqslant \lfloor|T|/2\rfloor$ (because $|T|=|S|+|I|$, and $|I|$ and $|T|$ are integers).	
	
To show that $|I|\leqslant |\mathcal{S}|-2$ we use a counting argument similar to that of \cite[Lemma 3.8]{Dillencourt1990}. Let $\mathcal{T}[\mathcal{S}]$ be the subgraph of $\mathcal{T}$ that is induced by $\mathcal{S}$. In other words, $\mathcal{T}[\mathcal{S}]$ is the resulting graph after removing vertices of $I$ and their incident edges from $\mathcal{T}$. 
 Since $\mathcal{T}$ is a triangulation and $I$ does not contain boundary vertices of $\mathcal{T}$, the removal of every vertex of $I$ creates a hole (a new face which is the union of original faces) whose boundary is a simple polygon. All edges of this polygon belong to $\mathcal{T}[\mathcal{S}]$ because $I$ is an independent set. Therefore, $\mathcal{T}[\mathcal{S}]$ is a connected plane graph, the boundaries of its interior faces are simple polygons, and the boundary of its outer face is the triangle $(u,v,w)$; see Figure~\ref{DT-fig}(b). Each interior face of $\mathcal{T}[\mathcal{S}]$ contains either no point of $I$ or exactly one point of $I$. Interior faces that do not contain any point of $I$ are called {\em good faces}, and interior faces that contain a point of $I$ are called {\em bad faces}. Each good face is a triangle. Let $g$ and $b$ denote the number of good and bad faces, respectively. Thus the number of interior faces is $g+b$. 
 
 Since $|I|=b$, it suffices to show that $b\leqslant |\mathcal{S}|-2$. To do so, we assign to each edge $(p,q)\in \mathcal{T}[\mathcal{S}]$ certain {\em distinguished angles}. If $(p,q)$ is an interior edge then we distinguish the two angles of $\mathcal{T}$ that are opposite to $(p,q)$, and if $(p,q)$ is a boundary edge then we distinguish the unique angle of $\mathcal{T}$ that is opposite to $(p,q)$, as in Figure~\ref{DT-fig}(b). Let $d$ be the total measure of all distinguished angles. We compute $d$ in two different ways: once with respect to the number of faces of $\mathcal{T}[\mathcal{S}]$ and once with respect to the number of edges of $\mathcal{T}[\mathcal{S}]$. Each good face contains three distinguished angles, their sum is $180^\circ$.
 The sum of the distinguished angles in each bad face is $360^\circ$ because these angles are anchored at the removed vertex in the face. Therefore 
 \begin{equation}
 \label{eq3}
 d=180\cdot g+360\cdot b.
 \end{equation}

Now we compute $d$ with respect to the number of edges of $\mathcal{T}[\mathcal{S}]$ which we denote by $e$. By Euler's formula, we have $e=|\mathcal{S}|+b+g-1$. By Inequality~\eqref{eq0}, the sum of (at most two) distinguished angles assigned to each edge is less than $180^\circ$. Therefore 
\begin{equation}
\label{eq2}
d< 180\cdot e= 180\cdot (|\mathcal{S}|+ b + g -1).
\end{equation} 

Combining \eqref{eq3} and \eqref{eq2}, we have
$$180\cdot g + 360\cdot b < 180\cdot (|\mathcal{S}|+ b + g -1),$$
which simplifies to $b< |\mathcal{S}|-1$. Since $b$ and $|\mathcal{S}|$ are integers, $b\leqslant |\mathcal{S}|-2$. 
\end{proof}

Our proof of Theorem~\ref{Dillencourt-thr} employs Theorem~\ref{independent-set-thr} and the following structural property of Delaunay triangulations presented by the author \cite{Biniaz2019}. For the sake of completeness we repeat its proof.

\begin{theorem}
	\label{Delaunay-thr}
		Let $T$ be the Delaunay triangulation of a set of points in the plane
		in general position. Let $p$ and $q$ be two vertices of $T$ and let $D$ be any closed disk that has on its boundary only vertices $p$ and $q$. Then there exists a path, between $p$ and $q$ in $T$, that lies in $D$.
\end{theorem}

\begin{proof}
	The proof is by induction on the number of vertices in $D$. If there is no vertex of $V(T)\setminus\{p,q\}$ in the interior of $D$, then $(p,q)$ is an edge of $T$, and so is a desired path. Assume that there exists a vertex $r\in V(T)\setminus\{p,q\}$ in the interior of $D$. Let $c$ be the center of $D$. Consider the ray $\ray{pc}$ emanating from $p$ and passing through $c$. Fix $D$ at $p$ and then shrink it along $\ray{pc}$ until $r$ lies on its boundary; see Figure \ref{DT-fig}(c). Denote the resulting disk $D_{pr}$, and notice that it lies fully in $D$. Compute the disk $D_{qr}$ in a similar fashion by shrinking $D$ along $\ray{qc}$. The disk $D_{pr}$ does not contain $q$ and the disk $D_{qr}$ does not contain $p$. By induction hypothesis there exists a path, between $p$ and $r$ in $T$, that lies in $D_{pr}$, and similarly there exists a path, between $q$ and $r$ in $T$, that lies in $D_{qr}$. The union of these two paths contains a path, between $p$ and $q$ in $T$, that lies in $D$. 
\end{proof}

\section{Proof of Theorem~\ref{Dillencourt-thr}}
\label{proof-section}

Recall $T$ and $S$. Pick an arbitrary representative vertex from each component of $T\setminus S$, and let $C$ be the set of these vertices. The number of components is $|C|$. Consider the Delaunay triangulation $T'$ of $S\cup C$. Observe that $C$ is an independent set of $T$. We prove by contradiction that $C$ is also an independent set of $T'$. Assume that there exists an edge $(c_1,c_2)\in T'$ such that $c_1,c_2\in C$. Since $T'$ is a Delaunay triangulation, by definition there exists a closed disk $D$ that has only $c_1$ and $c_2$ on its boundary and does not contain any other point of $S\cup C$. Now consider $T$ and $D$. By Theorem~\ref{Delaunay-thr} there exists a path between $c_1$ and $c_2$ in $T$, that lies in $D$. Since $D$ does not contain any point of $S$, all edges of this path belong to $T\setminus S$. This contradicts the fact that $c_1$ and $c_2$ belong to different components of $T\setminus S$. Therefore $C$ is an independent set of $T'$. 
By Theorem~\ref{independent-set-thr}, we have $|C|\leqslant |T'|/2$. This and the fact that $|T'|=|S|+|C|$ imply that $|C|\leqslant |S|$.   

\section{Blocking Delaunay triangulations} 
In this section, we use Theorem~\ref{independent-set-thr} and prove the conjecture of Aichholzer~\etal~\cite{Aichholzer2013} that at least $n$ points are required to block any $n$-vertex Delaunay triangulation.
Let $P$ be a set of points in the plane and let $T$ be the Delaunay triangulation of $P$. A point set $B$ {\em blocks} or {\em stabs} $T$ if in the Delaunay triangulation of $P\cup B$ there is no edge between two points of $P$. In other words, every disk that introduces an edge in $T$ contains a point of $B$. Throughout this section we assume that $P\cup B$ is in general position. 

In 2010, Aronov~\etal~\cite{Aronov2011} showed that $2n$ points are sufficient to block any $n$-vertex Delaunay triangulation, and if the vertices are in convex position then $4n/3$ points suffice. These bounds have been improved by Aichholzer~\etal~\cite{Aichholzer2013} (2010) to $3n/2$ and $5n/4$, respectively.

For the lower bound, Aronov~\etal~\cite{Aronov2011} showed the existence of $n$-vertex Delaunay triangulations that require $n$ points to be blocked, for example see Figure~\ref{sufficiency-fig}(a) in which every disk (representing a Delaunay edge) requires a unique point to be blocked as the disks are interior disjoint. Aichholzer~\etal~\cite{Aichholzer2013} proved that at least $n-1$ points are necessary to block any $n$-vertex Delaunay triangulations, and stated the following conjecture.

\begin{conjecture}
	\label{Aichholzer-conj}
	For any point set $P$ in the plane in convex position, $|P|$ points are necessary and sufficient to block the Delaunay triangulation of $P$.
\end{conjecture}

An implication of Theorem~\ref{independent-set-thr} proves the necessity of $|P|$ blocking points in Conjecture~\ref{Aichholzer-conj} (even if $P$ is in general position); the sufficiency remains open.

\begin{figure}[htb]
	\centering
	\setlength{\tabcolsep}{0in}
	$\begin{tabular}{cc}
	\multicolumn{1}{m{.5\columnwidth}}{\centering\includegraphics[width=.25\columnwidth]{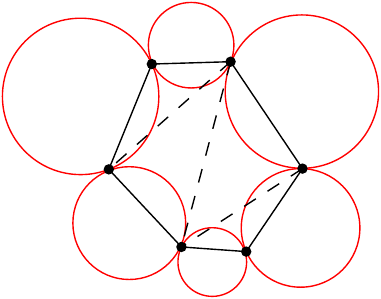}}
	&\multicolumn{1}{m{.5\columnwidth}}{\centering\includegraphics[width=.27\columnwidth]{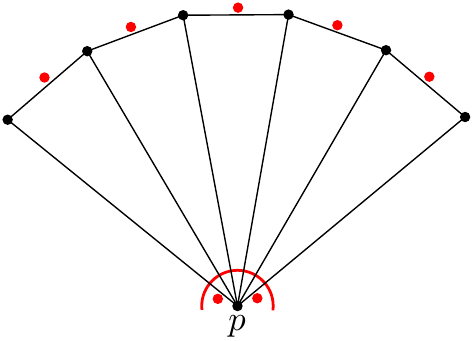}}\\
	(a)&(b) 
	\end{tabular}$
	\caption{(a) At least $n$ points are required to block this $n$-vertex Delaunay triangulation. (b) This $n$-vertex Delaunay triangulation can be blocked by $n$ points.}
	\label{sufficiency-fig}
\end{figure}

\begin{theorem}
	\label{blocking-thr}
	Let $P\cup B$ be any set of points in the plane in general position such that $B$ blocks the Delaunay triangulation of $P$. Then $|B|\geqslant |P|$, and this bound is tight. 
\end{theorem} 

\begin{proof}Consider the Delaunay triangulation $T$ of $P\cup B$. Since $B$ blocks the Delaunay triangulation of $P$, the removal of $B$ from $T$ leaves exactly $|P|$ components each consisting of a single point of $P$. Thus $P$ is an independent set of $T$. By Theorem~\ref{independent-set-thr}, we have $|P|\leqslant \lfloor|T|/2\rfloor\leqslant |T|/2$ which implies that $|B|\geqslant |P|$ (because $|T|=|P|+|B|$). 
	
To verify the tightness of this bound, consider a set of $n$ points in convex position where $n-1$ points are at distances approximately 1 from one point, say $p$, so that no four points lie on a circle. In the Delaunay triangulation of this point set, $p$ is connected to all other points, as depicted in Figure~\ref{sufficiency-fig}(b). This Delaunay triangulation can be blocked by $n$ points that are placed outside the convex hull: two points are placed very close to $p$ and $n-2$ points are placed very close to the $n-2$ convex hull edges that are not incident to $p$. A similar placement has also been used in \cite{Aichholzer2013} and \cite{Aronov2011}.
\end{proof}

\bibliographystyle{abbrv}
\bibliography{Toughness-DT}

\begin{thebibliography}{1}

\bibitem{Aichholzer2013}
O.~Aichholzer, R.~F. Monroy, T.~Hackl, M.~J. van Kreveld, A.~Pilz, P.~Ramos,
  and B.~Vogtenhuber.
\newblock Blocking {D}elaunay triangulations.
\newblock {\em Computational Geometry: Theory and Applications},
  46(2):154--159, 2013.
\newblock Also in {\em CCCG'10}.

\bibitem{Aronov2011}
B.~Aronov, M.~Dulieu, and F.~Hurtado.
\newblock Witness ({D}elaunay) graphs.
\newblock {\em Computational Geometry: Theory and Applications},
  44(6-7):329--344, 2011.
\newblock Also in arXiv:1008.1053, 2010.

\bibitem{Biniaz2019}
A.~Biniaz.
\newblock Plane hop spanners for unit disk graphs: Simpler and better.
\newblock In {\em Algorithms and Data Structures Symposium $($WADS$)$}, 2019.
\newblock Full version in arXiv:1902.10051.

\bibitem{Dillencourt1990}
M.~B. Dillencourt.
\newblock Toughness and {D}elaunay triangulations.
\newblock {\em Discrete $\&$ Computational Geometry}, 5:575--601, 1990.
\newblock Also in {\em SoCG'87}.

\bibitem{Tutte1947}
W.~T. Tutte.
\newblock The factorization of linear graphs.
\newblock {\em Journal of the London Mathematical Society}, 22:107--111, 1947.

\end{thebibliography}
\end{document}